\documentclass[12pt]{article}
\usepackage{amsfonts}
\usepackage{amsmath}
\usepackage{amssymb}
\usepackage{color}
\usepackage{eurosym}

\newtheorem{theorem}{Theorem}

\newtheorem{proposition}{Proposition}
\newtheorem{corollary}{Corollary}

\newtheorem{example}{Example}

\newenvironment{proof}[1][Proof]{\noindent\textbf{#1.} }{\ \rule{0.5em}{0.5em}}

\begin{document}

\title{The Shapley index for music streaming platforms\thanks{%
Financial support from grants PID2020-113440GBI00 and
PID2023-146364NB-I00, funded by MCIN/AEI/
10.13039/501100011033 and MICIU/AEI/10.13039/501100011033/
respectively, and by FEDER, UE, and grant ED431B2022/03 funded by Xunta de Galicia is
gratefully acknowledged.}}
\author{\textbf{Gustavo Berganti\~{n}os}\thanks{%
ECOBAS, Universidade de Vigo, ECOSOT, 36310 Vigo, Espa\~{n}a} \\
\textbf{Juan D. Moreno-Ternero}\thanks{%
Department of Economics, Universidad Pablo de Olavide, 41013 Sevilla, Espa%
\~{n}a; jdmoreno@upo.es}}
\maketitle

\begin{abstract}
We study an index to measure the popularity of artists in music streaming
platforms. This index, which can be used to allocate the amount raised via
paid subscriptions among participating artists, is based on the Shapley
value, a centerpiece in cooperative game theory. We characterize this 
\textit{Shapley index} combining several axioms formalizing principles with
normative appeal. This permits to place the index in the literature, as an
alternative to the well-known (and widely used in the industry) pro-rata and
user-centric indices.
\end{abstract}

\newpage

\section{Introduction}

Platform businesses have gained enormous attention in recent years, which
has been reflected into the literature on economics research (e.g., Cabral
et al., 2019; 
Belleflamme and Peitz, 2021; Calveras and Ganuza, 2021; Jullien et al., 2021). Among
other things, platforms have transformed the ways in which cultural content
is produced and consumed (e.g., Aguiar et al., 2024).\footnote{%
Nieborg and Poell (2018) have coined the term \textit{platformization of
cultural production}.} This is particularly the case with music. In the old
times, consumers typically learned about music from radio stations (or word
of mouth) and eventually moved on to buy from record stores. This started to
change when digital music emerged and spread
universally. After some initial years in which file-sharing platforms, such
as Napster, were under scrutiny by the music industry, Apple managed to
persuade record companies to sell individual tracks for 99 cents. Gradually,
the industry found new profitable paths, eventually embracing streaming, a
massive success nowadays. 
To wit, according to Statista, in the second quarter of 2024, Spotify (the
largest music streaming platform) reached an all-time high with 626 million
active users worldwide. This marked an increase of 75 million users in just
one year. 

Platforms cash such a massive success of streaming in various ways. But it
is estimated that almost 90\% of the total revenue that platforms raise
comes from \textit{premium consumers}. 
That is, consumers that gain access to all the music on the platform after
paying a monthly subscription.\footnote{%
Nevertheless, more than half of platform users do not pay any money (instead
listening to ads, while using the platform). 
The hybrid approach, offering both an ad-supported free version and an
ad-free subscription version is not exclusive of music platforms. Netflix,
Hulu, YouTube, or Pandora, to name a few, have also adopted it to deal with
the trade-off between viewership and subscription profits. An interesting
question, which we shall not study here, is how those multiple versions
should be designed and priced (e.g., Goli et al., 2024).} 
It is estimated that streaming platforms redistribute among artists around
65-70\% of the revenue they raise, which thus becomes a major aspect in the
management of streaming platforms.\footnote{%
To be more precise, streaming platforms pay ``right holders", who may be the
artists themselves if they are independent, or the record labels if the
artist is signed to one.} 
The problem of sharing the revenue raised from paid subscriptions to
streaming platforms among artists 
is a new form of revenue sharing problems under bundled pricing (e.g., Adams
and Yellen, 1976; Ginsburgh and Zang, 2003; Berganti\~{n}os and
Moreno-Ternero, 2015). As such, it offers new insights with respect to the
classical literature on industrial organization (e.g., Belleflamme and
Peitz, 2015).

In the early years, platforms used the \textit{pro-rata} method,
in which artists were rewarded in proportion of their total streams.
Gradually, they have been moving to a \textit{user-centric} method, in
which, instead, the amount paid by each user is shared among the artists
this user streamed, in proportion of the user's overall streams. The two
methods have been recently scrutinized in the scientific literature (e.g.,
Alaei et al., 2022; Berganti\~{n}os and Moreno-Ternero, 2024).

In this paper, we study a third method, which 
is obtained following the tradition of analyzing problems involving agents'
cooperation with a game-theoretical approach.\footnote{%
Classical instances are bankruptcy problems from the Talmud (e.g., Aumann
and Maschler, 1985), cost alocation problems (e.g., Tijs and Driessen,
1986), river sharing (e.g., Ambec and Sprumont, 2002), allocating benefits
of horizontal cooperation (e.g., Lozano et al., 2013), or the value captured
in hierarchical chains (e.g., Henkel and Hoffmann, 2018).} More precisely,
as in Berganti\~{n}os and Moreno-Ternero (2024), we associate to each
streaming problem (to be understood as the problem of allocating the overall
amount raised from paid subscriptions among artists streamed in the
platform) a cooperative (TU) game in which the worth of each coalition of artists is determined by the amount users streaming only those artists pay. We then consider the well-known Shapley value
(e.g., Shapley, 1953) of the resulting game as an allocation rule for
streaming problems.\footnote{%
Schlicher et al. (2024) associate another cooperative game to a streaming
problem. Both Berganti\~{n}os and Moreno-Ternero (2024) and Schlicher et al.
(2024) are concerned with the core of the resulting games, rather than the
Shapley value of such games. Gon\c{c}alves-Dosantos et al. (2024a) point out
that each of their indicators for streaming (not necessarily music)
platforms coincide with the Shapley value of different cooperative games
that can suitably be associated to the class of problems they analyze.} It
turns out that the Shapley value of such a TU game can be easily computed
(which, in other settings, is not always the case). In words, it says that the subscription of each user is equally (and fully)
allocated among the artists this user streamed. The ensuing allocation rule,
which will be the object of our study, is what we dub 
the \textit{Shapley index} for streaming problems.

The Shapley index is closer to the user-centric index than to the pro-rata
index mentioned above, as it also imposes that the amount each user pays is
distributed only among artists streamed by such a user. Now, the Shapley
index states that it is equally distributed among them, whereas the
user-centric says that it is proportionally distributed among them. In that
sense, both indices represent the two long-standing (and widely supported)
principles of distributive justice: egalitarianism and proportionality
(e.g., Young, 1994; Moulin, 2004; Thomson, 2019).

Beyond some preliminary game-theoretical results, which yield interesting
features for the Shapley index, we concentrate on its normative foundations.
More precisely, following the tradition initiated by Nash (1950) and Arrow
(1951), we take an axiomatic approach to streaming problems.\footnote{%
Berganti\~{n}os and Moreno-Ternero (2024) and Gon\c{c}alves-Dosantos et al.
(2024a, 2024b) have also applied the axiomatic approach to streaming
problems. Other recent instances of this approach, dealing with various
problems, are Asheim et al., (2020), Flores-Szwagrzak and Treibich (2020)
and Cs\'{o}ka and Herings (2021).} To do so, we formalize several principles
with normative appeal, referring to operational or ethical aspects of
streaming problems, as axioms of indices. We show that several combinations
of these axioms characterize the Shapley index. This permits a more thorough
comparison between this index and the other two main indices that existed so
far.

The rest of the paper is organized as follows. In Section 2, we present the
preliminaries of the model to analyze streaming problems. In Section 3, we
present the Shapley index and some game-theoretical aspects of it. In
Section 4, we present our axiomatic analysis. In Section 5, and based on the
results from the axiomatic analysis, we properly place the Shapley index in
the literature. Finally, Section 6 concludes. Some extra material (mostly
referring to the tightness of our characterization results) is gathered in
the appendix.

\section{Preliminaries}

We consider the model introduced in Berganti\~{n}os and Moreno-Ternero
(2024). Let $\mathbb{N}$ represent the set of all potential artists and $%
\mathbb{M}$ the set of all potential users (of music streaming platforms).
We can assume, without loss of generality, that both $\mathbb{N}$ and $%
\mathbb{M}$ are sufficiently large. 
Each specific platform involves a specific (finite) set of artists $N\subset 
\mathbb{N}$ and a specific (finite) set of users $M\subset \mathbb{M}$, with
cardinalities $n$ and $m$, respectively. For ease of notation, we typically
assume that $N=\left\{ 1,...,n\right\} $ and $M=\left\{ 1,...,m\right\} $.
For each pair $i\in N,j\in M$, let $t_{ij}$ denote the times user $j$ played
(via streaming) contents uploaded by artist $i$ in a platform (briefly, 
\textit{streams}), during a certain period of time (e.g., month). 
Let $t=\left( t_{ij}\right) _{i\in N,j\in M}$ denote the corresponding
matrix encompassing all streams. We assume that for each $j\in M,$ $%
\sum\limits_{i\in N}t_{ij}>0$ (namely, each user has streamed some content). 

A \textbf{streaming problem} is a triple $P=\left( N,M,t\right) $. We
normalize the amount paid by each user to $1$. Thus, the amount to be
divided among artists in a problem $\left( N,M,t\right) $ is just $m$, the
number of users. The set of problems so defined is denoted by $\mathcal{P}$.

For each $i\in N,$ we denote by $t_{-i}$ the matrix obtained from $t$ by
removing the row corresponding to artist $i$. Likewise, for each $j\in M$,
we denote by $t^{-j}$ the matrix obtained from $t$ by removing the column
corresponding to user $j$. For each artist $i\in N$, 
$T_{i}\left( N,M,t\right) =\sum_{j\in M}t_{ij}$, 
denotes the total times $i$ was streamed. Likewise, for each user $j\in M$, 
$T^{j}\left( N,M,t\right) =\sum_{i\in N}t_{ij}$, 
denotes the total times $j$ streamed content. Notice that, by assumption, $%
T^{j}\left( N,M,t\right) >0$.

We define the set of fans of each artist as the set of users who have
streamed content from the artist at least once. Formally, for each $i\in N,$ 
$F_{i}\left( N,M,t\right) =\left\{ j\in M:t_{ij}>0\right\}$. 
Similarly, we define the list of artists of a user as those from which the
user has streamed content at least once. Formally, for each $j\in M,$ 
$L^{j}\left( N,M,t\right) =\left\{ i\in N:t_{ij}>0\right\}$. 
The profile of user $j$ is defined as the streaming vector associated to
such a user. Namely, 
$t_{.j}\left( N,M,t\right) =\left( t_{ij}\right) _{i\in N}$. 
When no confusion arises we write $T_{i}$ instead of $T_{i}\left(
N,M,t\right) $, $T^{j}$ instead of $T^{j}\left( N,M,t\right) ,$ $F_{i}$
instead of $F_{i}\left( N,M,t\right) ,$ $L^{j}$ instead of $L^{j}\left(
N,M,t\right) ,$ and $t_{.j}$ instead of $t_{.j}\left( N,M,t\right) .$

A popularity \textbf{index} $\left( I\right) $ for streaming problems is a
mapping that measures the importance of each artist in each problem.
Formally, for each problem $\left( N,M,t\right) \in \mathcal{P}$, $I\left(
N,M,t\right) \in \mathbb{R}_{+}^{n}$ and, for each pair $i,j\in N$, $%
I_{i}\left( N,M,t\right) \geq I_{j}\left( N,M,t\right) $ if and only if $i$
is at least as important as $j$ at problem $\left( N,M,t\right) $. 
We assume that $\sum\limits_{i\in N}I_{i}\left( N,M,t\right) >0$. 

The reward received by each artist $i\in N$ from the revenues generated in
each problem ($m$ because the amount paid by each user has been normalized
to 1) is based on the importance of that artist in that problem. Formally, 
\begin{equation*}
R_{i}^{I}\left( N,M,t\right) =\frac{I_{i}\left( N,M,t\right) }{%
\sum\limits_{k\in N}I_{k}\left( N,M,t\right) }m.
\end{equation*}

Note that any positive linear transformation of a given index generates the
same allocation of rewards. Formally, for each $\lambda>0$ and each index $I$%
, $R^{\lambda I}\equiv R^{I}$. Thus, unless stated otherwise, 
we shall slightly abuse language to identify an index with all its positive
linear transformations.

Note also that our analysis does not impose that indices are normalized. That is, we do not impose that the aggregate amount indices yield for a problem is fixed (we only assume it is strictly positive, as mentioned above). Alternatively, if one imposes that indices are \textit{budget balanced}, i.e., $\sum\limits_{i\in N}I_{i}\left( N,M,t\right) =m$, then the index would coincide with its associated allocation rule, i.e., $I\equiv R^{I}$.

\section{The Shapley index}

Berganti\~{n}os and Moreno-Ternero (2024) associate with each streaming
problem $\left( N,M,t\right) \in \mathcal{P}$ a $TU$ game $\left(
N,v_{\left( N,M,t\right) }\right) $ where the set of agents of the
cooperative game is the set of artists.\footnote{%
A cooperative game with transferable utility, briefly a \textit{TU game}, is
a pair $\left( N,v\right) $, where $N$ denotes a set of agents and the value
function $v:2^{N}\rightarrow \mathbb{R}$, that satisfies $v\left(
\varnothing \right) =0$, yields the worth of each coalition of agents within 
$N$.} For each $S\subset N$, $v_{\left( N,M,t\right) }\left( S\right) $ is
defined as the amount paid by the users that have only streamed artists in $%
S $. Formally, 
\begin{equation*}
v_{\left( N,M,t\right) }\left( S\right) =\left\vert \left\{ j\in
M:L^{j}\subset S\right\} \right\vert .
\end{equation*}




The \textit{Shapley value} (Shapley, 1953) is the most well-known solution
concept for TU games. It is defined for each player as the average of his
contributions across orders of agents. Formally, for each $i\in N$, 
\begin{equation*}
Sh_{i}\left( N,v\right) =\frac{1}{n!}\sum_{\pi \in \Pi _{N}}\left[ v\left(
Pre\left( i,\pi \right) \cup \left\{ i\right\} \right) -v\left( Pre\left(
i,\pi \right) \right) \right] .
\end{equation*}

It turns out that, 
as stated in the next result, the Shapley value of the TU game associated to
a streaming problem can be easily computed.\footnote{%
The proof is similar to the proof of Proposition 1 in Ginsburgh and Zang
(2003) and, thus, we omit it here.}

\begin{proposition}
\label{sh formula} For each 
$\left( N,M,t\right)\in\mathcal{P} $ and each $i\in N,$ 
\begin{equation*}
Sh_{i}\left( N,v_{\left( N,M,t\right) }\right) =\sum_{j\in M}\sum_{i\in
L^{j}}\frac{1}{\left\vert L^{j}\right\vert }.
\end{equation*}
\end{proposition}

In words, the Shapley value of the TU game associated to a streaming problem
says that the subscription of each user $j$ is equally (and fully) allocated
among the artists $j$ streamed. The ensuing popularity index, which will be
the object of our study, is what we dub the \textbf{Shapley index} for
streaming problems. Formally, for each $\left( N,M,t\right)\in\mathcal{P}$,
and each $i\in N$, 
\begin{equation*}
Sh_{i}\left( N,M,t\right) =\sum_{j\in M}\sum_{i\in L^{j}}\frac{1}{\left\vert
L^{j}\right\vert }.
\end{equation*}





The definition of the TU game associated to streaming problems mentioned
above takes a \textit{pessimistic} stance. To wit, 
$v_{\left( N,M,t\right) }\left( S\right) $ says than only users that
streamed only artists in $S$ will count. 
Even if $99\%$ of the
streams of a user belong to artists in $S$, we do not consider this user in
the computation of $v_{\left( N,M,t\right) }\left( S\right) $.
Alternatively, one could consider an \textit{optimistic} stance, in which 
$v_{\left( N,M,t\right) }^{o}\left( S\right)$
would count all users that streamed some artist in $S$, even if only $1\%$ of streams of a user belong to artists in $S$.
\footnote{%
The pessimistic stance for the cooperative approach to streaming problems
is standard in the literature. It has been applied, among others, for
bankruptcy problems (e.g., O`Neill, 1982), minimum cost spanning tree
problems (e.g., Bird, 1976), museum pass problems (e.g., Ginsburgh and Zang,
2003), knapsack problems (e.g., Kellerer et al., 2004), or broadcasting
problems (e.g., Berganti\~{n}os and Moreno-Ternero, 2020). But in some of
those problems, an optimistic stance has also been endorsed sometimes. 
This has been the case, for instance, in bankruptcy problems (e.g., Driesen,
1995), minimum cost spanning tree problems (e.g., Berganti\~{n}os and
Vidal-Puga, 2007), and knapsack problems (e.g., Arribillaga and Berganti\~{n}%
os, 2022). See Atay and Trudeau (2024) for a general approach to both (optimistic and pessimistic) approches to cooperative games.} Formally, 
for each $S\subset N$ we define $v_{\left( N,M,t\right) }^{o}\left( S\right) 
$ as the amount paid by the users that have streamed artists in $S$. That
is, 
\begin{equation*}
v_{\left( N,M,t\right) }^{o}\left( S\right) =\left\vert \bigcup\limits_{i\in
S}F_{i}\left( N,M,t\right) \right\vert .
\end{equation*}




As stated in the next proposition, the pessimistic and optimistic games
associated to streaming problems are dual games, and, consequently, their
Shapley values coincide.\footnote{Driessen (1995) proved that, in
bankruptcy problems, the pessimistic and the optimistic game are also dual.
Nevertheless, this does not happen, for instance, for minimum cost spanning
tree problems or knapsack problems.} Formally, for each $TU$ game $\left( N,v\right) $, its \textit{dual}, denoted $\left(N,v^{\ast }\right) $, is defined by setting for each $S\subset N$, $v^{\ast}\left( S\right) =v\left( N\right) -v\left( N\backslash S\right) $. The worth $v^{\ast }(S)$ represents the amount that the other agents $N\setminus S$ cannot prevent $S$ from obtaining in $v$. 

\begin{proposition}
\label{optim game} For each $\left( N,M,t\right)\in\mathcal{P}$, its
associated $TU$ games, $\left( N,v_{\left( N,M,t\right) }\right) $ and $%
\left( N,v_{\left( N,M,t\right) }^{o}\right)$, are dual.
\end{proposition}

\begin{proof}
Let $\left( N,M,t\right) \in \mathcal{P}$, $S\subset N$, and $v_{\left(
N,M,t\right) }$ and $v_{\left( N,M,t\right) }^{o}$ defined as above. Then, 
\begin{eqnarray*}
v_{\left( N,M,t\right) }^{\ast }\left( S\right) &=&v_{\left( N,M,t\right)
}\left( N\right) -v_{\left( N,M,t\right) }\left( N\backslash S\right) \\
&=&m-\left\vert \left\{ j\in M:L^{j}\subset N\backslash S\right\} \right\vert
\\
&=&\left\vert \left\{ j\in M:L^{j}\cap S\neq \varnothing \right\} \right\vert
\\
&=&\left\vert \bigcup\limits_{i\in S}F_{i}\left( N,M,t\right) \right\vert \\
&=&v_{\left( N,M,t\right) }^{o}\left( S\right) .
\end{eqnarray*}
\end{proof}

\section{Axiomatic characterizations}


We now introduce some axioms that reflect normatively appealing properties
for indices. 
The first axiom says that if we can divide a problem as the sum of two
smaller problems, then the solution to the original problem should be the
sum of the solutions to the two smaller problems. Formally,

\textbf{Additivity}. For each triple $\left( N,M^{1},t^{1}\right) ,\left(
N,M^{2},t^{2}\right) ,\left( N,M,t\right) \in \mathcal{P}$ such that $%
M=M^{1}\cup M^{2}$, $M^{1}\cap M^{2}=\varnothing ,$ $t_{ij}=t_{ij}^{1}$ when 
$j\in M^{1}$ and $t_{ij}=t_{ij}^{2}$ when $j\in M^{2},$ 
\begin{equation*}
I\left( N,M,t\right) =I\left( N,M^{1},t^{1}\right) +I\left(
N,M^{2},t^{2}\right) .
\end{equation*}

To second axiom says that given a set of users $C$ and the set of artists $A$
streamed by those users, the amount received by the artists in $A$ should
be, at least, the amount paid by users in $C$. Formally,

\textbf{Reasonable lower bound. }
For each $\left( N,M,t\right) \in \mathcal{P}$ and each $C\subset M$, let $%
L^{C}=\bigcup\limits_{j\in C}L^{j}.$ Then, 
\begin{equation*}
\sum_{i\in L^{C}}\frac{I_{i}\left( N,M,t\right) }{\sum\limits_{k\in
N}I_{k}\left( N,M,t\right) }m\geq \left\vert C\right\vert .
\end{equation*}

The third axiom says that, as all users pay the same, then all users should
have the same (global) impact on the index. Formally,

\textbf{Equal global impact of users. } For each $\left( N,M,t\right) \in 
\mathcal{P}$ and each pair $j,j^{\prime }\in M,$ 
\begin{equation*}
\sum_{i\in N}I_{i}\left( N,M\backslash \left\{ j\right\} ,t^{-j}\right)
=\sum_{i\in N}I_{i}\left( N,M\backslash \left\{ j^{\prime }\right\}
,t^{-j^{\prime }}\right) .
\end{equation*}

The previous three axioms were considered in Berganti\~{n}os and
Moreno-Ternero (2024). The next one was not. It says that if two artists are
streamed by the same group of users, then the index for both artists should
coincide. Formally,

\textbf{Symmetry on fans.} For each $\left( N,M,t\right) \in \mathcal{P}$
and each pair $i,i^{\prime }\in N$ such that $F_{i}=F_{i^{\prime }}$, 
\begin{equation*}
I_{i}\left( N,M,t\right) =I_{i^{\prime }}\left( N,M,t\right) .
\end{equation*}

As the next result states, the combination of the previous four axioms
characterizes the Shapley index.

\begin{theorem}
\label{char Shapley} An index satisfies additivity, reasonable lower bound,
equal global impact of users and symmetry on fans if and only if it is the
Shapley index.
\end{theorem}

\begin{proof}
It is straightforward to see that the Shapley index is a member of the
family of probabilistic indices introduced by Berganti\~{n}os and
Moreno-Ternero (2024). Thus, it satisfies \textit{additivity} and \textit{%
reasonable lower bound}. It also satisfies \textit{equal global impact of
users}. To show that, let $\left( N,M,t\right) \in \mathcal{P}$ and $j\in M$%
. Then, 
\begin{eqnarray*}
\sum_{i\in N}Sh_{i}\left( N,M\backslash \left\{ j\right\} ,t^{-j}\right)
&=&\sum_{i\in N}\sum_{k\in M\backslash \left\{ j\right\}}\sum_{i\in L^{k}}\frac{1%
}{\left\vert L^{k}\right\vert }
=\sum_{k\in M\backslash \left\{ j\right\}
}\sum_{i\in N}\sum_{i\in L^{k}}\frac{1}{\left\vert L^{k}\right\vert } \\
&=&\sum_{k\in M\backslash \left\{ j\right\} }1=m-1,
\end{eqnarray*}%
which does not depend on $j\in M$ and hence the axiom follows. As for 
\textit{symmetry on fans}, let $\left( N,M,t\right) \in \mathcal{P}$, and $i$%
, $i^{\prime }\in N$ such that $F_{i}=F_{i^{\prime }}.$ For each $j\in M$,
we know that $i\in L^{j}$ if and only if $i^{\prime }\in L^{j}.$ Thus, 
\begin{equation*}
Sh_{i}\left( N,M,t\right) =\sum_{j\in M}\sum_{i\in L^{j}}\frac{1}{\left\vert
L^{j}\right\vert }=\sum_{j\in M}\sum_{i^{\prime }\in L^{j}}\frac{1}{\left\vert
L^{j}\right\vert }=Sh_{i^{\prime }}\left( N,M,t\right) .
\end{equation*}

Conversely, let $I$ be an index satisfying the four axioms in the statement. 
Let $\left( N,M,t\right)\in\mathcal{P}$. 
Let $j,j^{\prime }\in M$ and 
$x_{j},x_{j^{\prime }}\in\mathbb{Z}^n_+$. 
By \textit{equal global impact of users}, 
\begin{eqnarray*}
\sum\limits_{i\in N}I_{i}\left( N,\left\{ j^{\prime }\right\} ,x_{j^{\prime
}}\right) &=&\sum_{i\in N}I_{i}\left( N,\{j,j^\prime \}\backslash \left\{
j\right\} ,x_{j^{\prime }}\right) \\
&=&\sum_{i\in N}I_{i}\left( N,\{j,j^\prime \}\backslash \left\{ j^{\prime
}\right\} ,x_{j}\right) \\
&=&\sum\limits_{i\in N}I_{i}\left( N,\left\{ j\right\} ,x_{j}\right) .
\end{eqnarray*}

Then, we can define $\lambda _{N}=\sum\limits_{i\in N}I_{i}\left( N,\left\{
j\right\} ,x_{j}\right)$. 

By \textit{additivity}, 
\begin{equation*}
I\left( N,M,t\right) =\sum_{j\in M}I\left( N,\left\{ j\right\}
,t_{.j}\right) .
\end{equation*}

Let $j\in M$. By assumption, $L^{j}\left( N,\left\{ j\right\},t_{.j}\right)\neq \varnothing$. Let $C=\left\{ j\right\}$. By 
\textit{reasonable lower bound}, 
\begin{equation*}
\sum_{i\in L^{j}\left( N,\left\{ j\right\} ,t_{.j}\right) }\frac{I_{i}\left(
N,\left\{ j\right\} ,t_{.j}\right) }{\sum\limits_{k\in N}I_{k}\left(
N,\left\{ j\right\} ,t_{.j}\right) }\geq 1.
\end{equation*}
That is, 
\begin{equation*}
\sum_{i\in L^{j}\left( N,\left\{ j\right\} ,t_{.j}\right) } I_{i}\left(N,\left\{ j\right\} ,t_{.j}\right) \geq \sum\limits_{i\in N}I_{i}\left(N,\left\{ j\right\} ,t_{.j}\right).
    \end{equation*}

As $I_{i}\left( N,\left\{ j\right\} ,t_{.j}\right)
\geq 0$ for each $i\in N$, we deduce that $I_{i}\left( N,\left\{ j\right\} ,t_{.j}\right)
=0$, for each $i\in N\backslash L^{j}\left( N,\left\{j\right\} ,t_{.j}\right)$.

By \textit{symmetry on fans}, for each pair $i,i^{\prime }\in L^{j}\left(
N,\left\{ j\right\} ,t_{.j}\right) ,$ 
\begin{equation*}
I_{i^{\prime }}\left( N,\left\{ j\right\} ,t_{.j}\right) =I_{i}\left(
N,\left\{ j\right\} ,t_{.j}\right) .
\end{equation*}

As $L^{j}\left( N,\left\{ j\right\} ,t_{.j}\right) =L^{j}\left( N,M,t\right)$%
, 
\begin{equation*}
I_{i}\left( N,\left\{ j\right\} ,t_{.j}\right) =\left\{ 
\begin{tabular}{ll}
$\frac{\lambda _{N}}{\left\vert L^{j}\left( N,M,t\right) \right\vert }$ & if 
$i\in L^{j}\left( N,M,t\right) $ \\ 
0 & otherwise.%
\end{tabular}%
\right.
\end{equation*}

Hence, 
\begin{eqnarray*}
I_{i}\left( N,M,t\right) &=&\sum_{j\in M}I_{i}\left( N,\left\{ j\right\}
,t_{.j}\right) =\sum_{j\in M:i\in L^{j}\left( N,M,t\right) }\frac{\lambda
_{N}}{\left\vert L^{j}\left( N,M,t\right) \right\vert } \\
&=&\lambda _{N}Sh_{i}\left( N,M,t\right) .
\end{eqnarray*}

Then, $I$ is a positive linear transformation of $Sh$, and thus $R^I\equiv
Sh $.
\end{proof}

\bigskip

The axioms in Theorem \ref{char Shapley} are independent (see the appendix).

\bigskip

As stated in the next result, \textit{symmetry on fans} can be replaced at
Theorem \ref{char Shapley} by the following pair of axioms. First, the axiom
of \textit{order preservation}, which formalizes in this context the
standard notion in the literature on resource allocation. That is, if each
user $j$ has streamed artist $i^{\prime }$ at least the same number of times
as artist $i$, then the index cannot yield for artist $i^{\prime }$ a
smaller number than for artist $i$. Formally,

\textbf{Order preservation}. For each $\left( N,M,t\right) \in \mathcal{P}$
and each pair $i,i^{\prime }\in N$, such that $t_{ij}\leq t_{i^{\prime }j}$
for all $j\in M$, 
\begin{equation*}
I_{i}\left( N,M,t\right) \leq I_{i^{\prime }}\left( N,M,t\right) .
\end{equation*}

Second, an axiom echoing a concern in the music industry, limiting the
options artists have to manipulate their streaming popularity. To wit, it
has been claimed that remuneration based on the number of streams may lead
artists to produce more and shorter songs (e.g., Meyn et al., 2023).%
\footnote{%
Overall song length has decreased by 2.5 seconds per year, $-10\%$ over the
last five years (e.g., Meyn et al., 2023). To add some anecdotical evidence,
``Stairway to Heaven", 
regarded as one of the greatest songs of all time and the most-requested
song on FM radio stations in the United States in the early 70's, is 8
minutes long. None of the Top 10 most streamed songs on Spotify (all of them
released at least 40 years later) surpasses half that length. And very few
do so within the Top 100. One notable exception is ``Bohemian Rapsody'',
almost 6 minutes long and released only four years later than ``Stairway to
Heaven".} 
Imagine, for instance, that an artist produces 4 songs with an average
duration of 3 minutes, instead of producing 3 songs with an average duration
of 4 minutes. A user devoting 12 minutes to stream this artist might yield
more profits in the first case than in the second case. 
Our model cannot fully accommodate this scenario, as we only consider
streams (neither songs, nor length). But we can nevertheless echo this
concern by presenting an axiom stating that an artist cannot increase the
index unilaterally upon increasing its number of streams from fans. Formally,

\bigskip

\textbf{Non-unilateral manipulability}. For each pair $\left( N,M,t\right)
,\left(N,M,t^{\prime }\right) \in \mathcal{P}$, and each $i\in N$ such that $%
F_{i}\left( N,M,t\right) =F_{i}\left( N,M,t^{\prime }\right) $, $t_{ij}\leq
t_{ij}^{\prime }$ for all $j\in M$, and $t_{kj}=t_{kj}^{\prime }$ for all $%
k\in N\backslash \left\{ i\right\} $ and $j\in M$, 
\begin{equation*}
I_{i}\left( N,M,t\right) \geq I_{i}\left( N,M,t^{\prime }\right) .
\end{equation*}

We can now present our second result.

\begin{theorem}
\label{char Shapley song lengh}An index satisfies additivity, reasonable
lower bound, equal global impact of users, order preservation, and
non-unilateral manipulability if and only if it is the Shapley index.
\end{theorem}

\begin{proof}
We have seen that the Shapley index satisfies \textit{additivity}, \textit{%
reasonable lower bound}, and \textit{equal global impact of users}. It is
straightforward to see that the Shapley index also satisfies \textit{order
preservation}, and \textit{non-unilateral manipulability}.

Conversely, let $I$ be an index satisfying the axioms in the statement. Let $%
j\in M.$ By \textit{equal global impact of users}, using similar arguments
to those used in the proof of Theorem \ref{char Shapley}, we can define $%
\lambda _{N}=\sum\limits_{i\in N}I_{i}\left( N,\left\{ j\right\} ,x\right) $
because it does not depend on $j$ or $x.$

By \textit{reasonable lower bound}, using similar arguments to those used
in the proof of Theorem \ref{char Shapley}, we obtain that for each $i\notin L^{j}\left(
N,\left\{ j\right\} ,x\right) ,$ $I_{i}\left( N,\left\{ j\right\} ,x\right)
=0.$ We now prove that for each $i\in L^{j}\left( N,\left\{ j\right\}
,x\right) ,$ 
\begin{equation*}
I_{i}\left( N,\left\{ j\right\} ,x\right) =\frac{\lambda _{N}}{\left\vert
L^{j}\left( N,\left\{ j\right\} ,x\right) \right\vert }.
\end{equation*}

The proof is by induction on 
\begin{equation*}
s=\left\vert \left\{ i\in N:x_{i}=\max_{i^{\prime }\in N}\left\{
x_{i^{\prime }}\right\} \right\} \right\vert .
\end{equation*}

Assume that $s=n.$ Then, $x_{i}=x_{i^{\prime }}>0$ for all $i,i^{\prime }\in
N.$ Notice that, by assumption, $L^{j}\left( N,\left\{ j\right\} ,x\right)
\neq \varnothing .$ By \textit{order preservation}, $I_{i}\left( N,\left\{
j\right\} ,x\right) =I_{i^{\prime }}\left( N,\left\{ j\right\} ,x\right) .$
Thus, for each $i\in L^{j}\left( N,\left\{ j\right\} ,x\right) =N,$ 
\begin{equation*}
I_{i}\left( N,\left\{ j\right\} ,x\right) =\frac{\lambda _{N}}{\left\vert
L^{j}\left( N,\left\{ j\right\} ,x\right) \right\vert }.
\end{equation*}

Assume now that the result holds whens $s\geq r$ and we prove it when $%
s=r-1. $ We consider two cases.

\begin{enumerate}
\item $L^{j}\left( N,\left\{ j\right\} ,x\right) =\left\{ i\in
N:x_{i}=\max\limits_{i^{\prime }\in N}\left\{ x_{i^{\prime }}\right\}
\right\}$. 

By \textit{order preservation}, $I_{i}\left( N,\left\{ j\right\}
,x\right) =I_{i^{\prime }}\left( N,\left\{ j\right\} ,x\right) $ for all $%
i,i^{\prime }\in L^{j}\left( N,\left\{ j\right\} ,x\right) .$ Thus, for each 
$i\in L^{j}\left( N,\left\{ j\right\} ,x\right) ,$ 
\begin{equation*}
I_{i}\left( N,\left\{ j\right\} ,x\right) =\frac{\sum\limits_{i\in
N}I_{i}\left( N,\left\{ j\right\} ,x\right) }{\left\vert L^{j}\left(
N,\left\{ j\right\} ,x\right) \right\vert }=\frac{\lambda _{N}}{\left\vert
L^{j}\left( N,\left\{ j\right\} ,x\right) \right\vert }.
\end{equation*}

\item $L^{j}\left( N,\left\{ j\right\} ,x\right) \neq \left\{ i\in
N:x_{i}=\max\limits_{i^{\prime }\in N}\left\{ x_{i^{\prime }}\right\}
\right\}$. 

Let $i\in $ $L^{j}\left( N,\left\{ j\right\} ,x\right)
\backslash \left\{ i\in N:x_{i}=\max\limits_{i^{\prime }\in N}\left\{
x_{i^{\prime }}\right\} \right\} .$ Let $x^{\prime }=\left( x_{i^{\prime
}}^{\prime }\right) _{i^{\prime }\in N}$ be such that $x_{i}^{\prime }=\max
\left\{ x_{i^{\prime }}\right\} _{i^{\prime }\in N}$ and $x_{i^{\prime
}}^{\prime }=x_{i^{\prime }}$ when $i^{\prime }\in N\backslash \left\{
i\right\} .$ Notice that $L^{j}\left( N,\left\{ j\right\} ,x\right)
=L^{j}\left( N,\left\{ j\right\} ,x^{\prime }\right) .$

By induction hypothesis, for each $i\in L^{j}\left( N,\left\{ j\right\}
,x\right) $, 
\begin{equation*}
I_{i}\left( N,\left\{ j\right\} ,x^{\prime }\right) =\frac{\lambda _{N}}{%
\left\vert L^{j}\left( N,\left\{ j\right\} ,x^{\prime }\right) \right\vert }.
\end{equation*}%
By \textit{non-unilateral manipulability}, 
\begin{equation*}
I_{i}\left( N,\left\{ j\right\} ,x^{\prime }\right) \leq I_{i}\left(
N,\left\{ j\right\} ,x\right) .
\end{equation*}%
By \textit{order preservation}, for each $i^{\prime }\in \left\{ i\in
N:x_{i}=\max\limits_{i^{\prime }\in N}\left\{ x_{i^{\prime }}\right\},
\right\} $ 
\begin{equation*}
I_{i}\left( N,\left\{ j\right\} ,x\right) \leq I_{i^{\prime }}\left(
N,\left\{ j\right\} ,x\right) .
\end{equation*}%
Thus, for each $i^{\prime }\in L^{j}\left( N,\left\{ j\right\} ,x\right) ,$ 
\begin{equation*}
I_{i^{\prime }}\left( N,\left\{ j\right\} ,x\right) \geq \frac{\lambda _{N}}{%
\left\vert L^{j}\left( N,\left\{ j\right\} ,x\right) \right\vert },
\end{equation*}%
and hence, for each $i^{\prime }\in L^{j}\left( N,\left\{ j\right\}
,x\right) ,$ 
\begin{equation*}
I_{i^{\prime }}\left( N,\left\{ j\right\} ,x\right) =\frac{\lambda _{N}}{%
\left\vert L^{j}\left( N,\left\{ j\right\} ,x\right) \right\vert }.
\end{equation*}
\end{enumerate}

Let $\left( N,M,t\right) \in P$. By \textit{additivity}, using similar
arguments to those used in the proof of Theorem \ref{char Shapley}, we
conclude that $I_{i}\left( N,M,t\right) =\lambda _{N}Sh_{i}\left(
N,M,t\right)$. Then, $I$ is a positive linear transformation of $Sh$, and
thus $R^{I}\equiv Sh$.
\end{proof}

\bigskip

The axioms in Theorem \ref{char Shapley song lengh} are independent (see the
appendix).

\bigskip

We now consider a new axiom of a different (variable-population) nature,
stating that if artist $i$ leaves the platform, the change in the index of
any other artist $i^{\prime }$ coincides with the change in the index to
artist $i$ when artist $i^{\prime }$ leaves the problem. Formally,

\textbf{Equal impact of artists}. For each $\left( N,M,t\right) \in \mathcal{%
P}$ and each pair $i,i^{\prime }\in N$, 
\begin{equation*}
I_{i}\left( N,M,t\right) -I_{i}\left( N\backslash \left\{ i^{\prime
}\right\} ,M,t_{-i^{\prime }}\right) =I_{i^{\prime }}\left( N,M,t\right)
-I_{i^{\prime }}\left( N\backslash \{i\},M,t_{-i}\right) .
\end{equation*}

It turns out that, as the next result states, we can also characterize the
Shapley index just replacing \textit{symmetry on fans} by this axiom at
Theorem \ref{char Shapley} (or \textit{order preservation} and \textit{%
non-unilateral manipulability} by this axiom at Theorem \ref{char Shapley
song lengh}).

\begin{theorem}
\label{char Shapley-vp} An index satisfies additivity, reasonable lower
bound, equal global impact of users and equal impact of artists if and only
if it is the Shapley index.
\end{theorem}

\begin{proof}
We know from Theorem \ref{char Shapley} that the Shapley index satisfies 
\textit{additivity}, \textit{reasonable lower bound} and \textit{equal
global impact of users}. As for \textit{equal impact of artists}, let $%
\left( N,M,t\right)\in\mathcal{P}$ and $i,i^{\prime }\in N$. To ease
notation, we let $P=\left( N,M,t\right)$, $P_{-i}=\left( N\backslash \left\{
i\right\} ,M,t_{-i}\right)$, and $P_{-i^{\prime }}=\left( N\backslash
\left\{ i^{\prime }\right\},M,t_{-i^{\prime }}\right)$. Then, 
\begin{eqnarray*}
I_{i}\left( P\right) -I_{i}\left( P_{-i^{\prime }}\right) &=&\sum_{j\in
M}\sum_{i\in L^{j}\left( P\right) }\frac{1}{\left\vert L^{j}\left( P\right)
\right\vert }-\sum_{j\in M}\sum_{i\in L^{j}\left( P_{-i^{\prime }}\right) }\frac{1%
}{\left\vert L^{j}\left( P_{-i^{\prime }}\right) \right\vert } \\
&=&\sum_{j\in M}\sum_{i,i^{\prime }\in L^{j}\left( P\right) }\left( \frac{1}{%
\left\vert L^{j}\left( P\right) \right\vert }-\frac{1}{\left\vert
L^{j}\left( P\right) \right\vert -1}\right) \\
&=&\sum_{j\in M}\sum_{i^{\prime }\in L^{j}\left( P\right) }\frac{1}{\left\vert
L^{j}\left( P\right) \right\vert }-\sum_{j\in M}\sum_{i^{\prime }\in L^{j}\left(
P_{-i}\right) }\frac{1}{\left\vert L^{j}\left( P_{-i}\right) \right\vert } \\
&=&I_{i^{\prime }}\left( P\right) -I_{i^{\prime }}\left( P_{-i}\right) .
\end{eqnarray*}

Conversely, let $I$ be an index that satisfies all the axioms in the
statement. We then prove that for each $\left( N,\left\{ j\right\} ,t\right)
\in\mathcal{P}$ there exists $\lambda \in \mathbb{R}_{+}$ such that 
\begin{equation}  \label{proof sh uni}
I_{i}\left( N,\left\{ j\right\} ,t\right) =\left\{ 
\begin{tabular}{ll}
$\frac{\lambda }{\left\vert L^{j}\left( N,\left\{ j\right\}
,t\right)\right\vert }$ & if $i\in L^{j}\left( N,\left\{ j\right\} ,t\right) 
$ \\ 
0 & otherwise.%
\end{tabular}%
\right.
\end{equation}

The proof is by induction on $n$, the number of artists.

The base case $n=1$ holds by taking $\lambda =\lambda _{\left\{ 1\right\} }$%
, as obtained in the first step of the proof of Theorem \ref{char Shapley}
above, which only required \textit{equal global impact of users}.

We assume now that $\left( \ref{proof sh uni}\right) $ holds for less than $%
n $ artists. We prove it for $n.$

As in the previous proofs, by \textit{reasonable lower bound}, $I_{i}\left( N,\left\{ j\right\}
,t\right) =0$, for each $i\in N$ such that $t_{ij}=0$. 
Let $i\in N$ be such that $t_{ij}>0$. We consider two cases:

\begin{enumerate}
\item There exists $i^{\prime }\in N$ such that $t_{i^{\prime }j}=0.$

By \textit{equal impact of artists}, 
\begin{equation*}
I_{i}\left( P\right) -I_{i}\left( P_{-i^{\prime }}\right) =I_{i^{\prime
}}\left( P\right) -I_{i^{\prime }}\left( P_{-i}\right) .
\end{equation*}%
As mentioned above, $I_{i^{\prime }}\left( P\right) =0$. By induction
hypothesis, $I_{i^{\prime }}\left( P_{-i}\right) =0$ and 
\begin{equation*}
I_{i}\left( P_{-i^{\prime }}\right) =\frac{\lambda }{\left\vert L^{j}\left(
P_{-i^{\prime }}\right) \right\vert }=\frac{\lambda }{\left\vert L^{j}\left(
P\right) \right\vert }.
\end{equation*}%
Thus, 
\begin{equation*}
I_{i}\left( P\right) =I_{i}\left( P_{-i^{\prime }}\right) =\frac{\lambda }{%
\left\vert L^{j}\left( P\right) \right\vert }.
\end{equation*}%
Hence, $\left( \ref{proof sh uni}\right) $ holds.

\item $L^{j}\left( P\right) =N.$

Let $i^{\prime }\in N\backslash \left\{ i\right\} $. By \textit{equal impact
of artists}, 
\begin{equation*}
I_{i}\left( P\right) -I_{i}\left( P_{-i^{\prime }}\right) =I_{i^{\prime
}}\left( P\right) -I_{i^{\prime }}\left( P_{-i}\right) .
\end{equation*}%
By induction hypothesis, 
\begin{eqnarray*}
I_{i^{\prime }}\left( P_{-i}\right) &=&\frac{\lambda }{\left\vert
L^{j}\left( P_{-i}\right) \right\vert }=\frac{\lambda }{\left\vert
L^{j}\left( P\right) \right\vert -1}\text{ and} \\
I_{i}\left( P_{-i^{\prime }}\right) &=&\frac{\lambda }{\left\vert
L^{j}\left( P_{-i^{\prime }}\right) \right\vert }=\frac{\lambda }{\left\vert
L^{j}\left( P\right) \right\vert -1}.
\end{eqnarray*}

Thus, 
\begin{equation*}
I_{i}\left( P\right) =I_{i^{\prime }}\left( P\right) .
\end{equation*}%
Hence, 
\begin{equation*}
I_{i}\left( P\right) =\frac{\sum\limits_{i^{\prime }\in N}I_{i^{\prime
}}\left( P\right) }{n}=\frac{\lambda _{N}}{n}.
\end{equation*}

It only remains to prove that $\lambda _{N}=\lambda .$ Given $i^{\prime }\in
N\ $we define $t^{\prime }$ as the vector obtained from $t$ by nullifying
the streams of artist $i^{\prime }.$ Namely, $t_{ij}^{\prime }=t_{ij}$ when $%
i\neq i^{\prime }$ and $t_{i^{\prime }j}=0.$ By Case 1, $I_{i^{\prime
}}\left( N,\left\{ j\right\} ,t^{\prime }\right) =0$ and $I_{i^{\prime
\prime }}\left( N,\left\{ j\right\} ,t^{\prime }\right) =\frac{\lambda }{n-1}
$ when $i^{\prime \prime }\neq i^{\prime }.$ By \textit{equal global impact
of users}, using similar arguments to those used in the proof of Theorem \ref%
{char Shapley}, we can define $\lambda _{N}=\sum\limits_{i\in N}I_{i}\left(
N,\left\{ j\right\} ,x\right) $ because it neither depends on $j$ nor on $x$.
Then, 
\begin{equation*}
\lambda _{N}=\sum\limits_{i^{\prime }\in N}I_{i^{\prime }}\left( N,\left\{
j\right\} ,t^{\prime }\right) =\lambda .
\end{equation*}
\end{enumerate}

Let $\left( N,M,t\right) \in P$. By \textit{additivity}, using similar
arguments to those used in the previous proofs, we
conclude that $I_{i}\left( N,M,t\right) =\lambda _{N}Sh_{i}\left(
N,M,t\right)$. Then, $I$ is a positive linear transformation of $Sh$, and
thus $R^{I}\equiv Sh$.
\end{proof}

\bigskip

The axioms in Theorem \ref{char Shapley-vp} are also independent (see the
appendix).

\bigskip

The Shapley index also satisfies the following natural axiom. It says that
if the number of streams of one artist is $0,$ then the importance of this
artist should be $0.$ Formally,

\textbf{Null artists.} For each $\left( N,M,t\right)\in\mathcal{P}$, and
each $i\in N$ such that $T_{i}=0$, 
\begin{equation*}
I_{i}\left( N,M,t\right) =0.
\end{equation*}


A close inspection of the proof of Theorems \ref{char Shapley}, \ref{char
Shapley song lengh}, and \ref{char Shapley-vp} allows us to state that 
\textit{null artists} can replace \textit{reasonable lower bound} in each of those
results to also characterize the Shapley index. That is,

\begin{corollary}
\label{char Shapley-corol} The following statements hold:

\begin{enumerate}
\item An index satisfies additivity, null artists, equal global impact of
users and symmetry on fans if and only if it is the Shapley index.

\item An index satisfies additivity, null artists, equal global impact of
users, order preservation and non-unilateral manipulability if and only if
it is the Shapley index.

\item An index satisfies additivity, null artists, equal global impact of
users and equal impact of artists if and only if it is the Shapley index.
\end{enumerate}
\end{corollary}

\bigskip

We conclude this section stressing that the characterizations of Theorem \ref%
{char Shapley} and Theorem \ref{char Shapley song lengh} and that of Theorem %
\ref{char Shapley-vp} are qualitatively different. The former only consider
fixed-population axioms, whereas the latter considers a variable-population
axiom (\textit{equal impact of artists}). 

\bigskip

\section{Placing the Shapley index in the literature}

The Shapley index we study in this paper is an alternative to two other
indices that have played a central role for streaming problems: the
so-called pro-rata and user-centric indices. The first one assigns to each
artist the overall number of streams. Consequently, it rewards artists
proportionally to their number of total streams. The second one assigns each
artist the result from allocating each user's subscription fee
proportionally among the artists streamed by that user and aggregating
afterwards.

\begin{equation*}
    \begin{tabular}{|l|c|c|c|}
    \hline
    & Shapley & Pro-rata & User-centric \\ \hline
    Additivity & Yes & Yes & Yes \\ \hline
    Reasonable lower bound & Yes & No & Yes \\ \hline
    Equal global impact of users & Yes & No & Yes \\ \hline
    Symmetry on fans & Yes & No & No \\ \hline
    Order preservation & Yes & Yes & Yes \\ \hline
    Non-unilateral manipulability & Yes & No & No \\ \hline
    Null artists & Yes & Yes & Yes \\ \hline
    Equal impact of artists & Yes & Yes & No \\ \hline
    Pairwise homogeneity & No & Yes & Yes \\ \hline
    Click-fraud-proofness & Yes & No & Yes \\ \hline
    \end{tabular}%
    \end{equation*}
    \begin{equation*}
    \text{Table 1. Rules and axioms}
    \end{equation*}

    
In order to compare the Shapley index with the other two, we resort to the
axiomatic analysis from previous section. Table 1 summarizes the
performance of the three indices with
respect to the axioms considered above, as well as two additional axioms (\textit{pairwise homogeneity} and \textit{click-fraud-proofness}) introduced by Berganti\~{n}os and Moreno-Ternero (2024).\footnote{\textit{Pairwise homogeneity} says that if each user streams an artist certain times more than another artist, the index should preserve that ratio. \textit{Click-fraud proofness} says that if a user changes streams, payments to artists cannot change more than the user's
subscription. It is straightforward to check that the Shapley index satisfies the latter, but not the former. It is also straightforward to check that pro-rata satisfies null artists, order
preservation and equal impact of artists, but neither symmetry on fans and
non-unilateral manipulability. And that
user-centric satisfies null artists and order preservation but not symmetry
on fans, non-unilateral manipulability, or equal impact of artists. For the
performance of these two indices with respect to rest of the axioms, the
reader is referred Berganti\~{n}os and Moreno-Ternero (2024).}

We can see from Table 1 that the Shapley index satisfies
almost all of the axioms. The only exception is pairwise homogeneity. This axiom is satisfied by the user-centric index, which on the other hand violates symmetry on fans, non-unilateral
manipulability and equal impact of artists. As a matter of fact, the user-centric index is the only
probabilistic index (Berganti\~{n}os and Moreno-Ternero, 2024) that
satisfies pairwise homogeneity, whereas the Shapley index is the only
probabilistic index that satisfies either symmetry on fans, equal impact of
artists or non-unilateral manipulability. If one believes that pairwise
homogeneity is a more appealing axiom than the other three, then one should
conclude that the user-centric rule behaves better than the Shapley index.
Otherwise, the conclusion should be reversed.

Comparing the pro-rata and the Shapley indices, we also observe that the former
satisfies pairwise homogeneity whereas the latter does not. On the other
hand, the latter satisfies reasonable lower bound, equal global impact of
users, symmetry on fans, non-unilateral manipulability and
click-fraud-proofness, whereas the former does not. We believe that
endorsing one index over the other should be a decision based on comparing
those sets of axioms.

We conclude comparing our
characterization results with other results in the literature. Gon\c{c}%
alves-Dosantos et al. (2024a) also characterize the Shapley index for streaming
problems (dubbed \textit{subscriber-uniform indicator} in their setting).\footnote{%
In their parlance, they work with \textit{indicators}, which can roughly be
interpreted as normalized indexes, which assign numbers to artist that add
up to exactly the overall amount users pay.} Their characterization is obtained by combining the axioms of composition (additivity in our setting), strong symmetry (symmetry on fans in our setting) and nullity (null artists in our
setting). This characterization can thus be seen as a counterpart of the original
characterization of the Shapley value for TU-games on the grounds of efficiency,
additivity, symmetry and dummy (Shapley, 1953).\footnote{Macho-Stadler et al. (2007) extend this result to environments with externalities.} Our Theorem \ref{char Shapley} makes use of two of those axioms (additivity
and symmetry on fans). The other two axioms (reasonable lower bound and
equal global impact of users) are different.\footnote{Theorem 1 in Berganti\~{n}os and Moreno-Ternero (2015) can also be seen as a counterpart of Shapley's original characterization in the so-called museum pass problem. Therein, equal treatment of equals, dummy, and additivity on visitors are considered (which are similar to symmetry on fans, null artists and additivity, respectively, in our setting) as well as efficiency (which is replaced
in Theorem \ref{char Shapley} by the axiom of equal global impact of users).} In Corollary 1.1, null artists replaces reasonable lower bound, thus sharing three axioms with Shapley's original characterization (replacing efficiency with
equal global impact of users).  

Theorem \ref{char Shapley song lengh} replaces symmetry on fans by order
preservation and non-unilateral manipulability (which is a new axiom,
tailor-suited for this context).

Theorem \ref{char Shapley-vp} is related to the characterization of the
Shapley value on the grounds of balanced contributions and efficiency (e.g.,
Myerson, 1980).\footnote{%
Counterpart results have also been obtained for other allocation problems,
where the Shapley value of the (optimistic) associated game is characterized
(e.g., Berganti\~{n}os and Vidal-Puga, 2007; Arribillaga and Berganti\~{n}%
os, 2022).} Our axiom of equal impact of artists is the counterpart of
Myerson's balanced contributions in our setting. Now, as our definition of
index does not include efficiency, we need three additional axioms to close
the characterization of the Shapley index. To the best of our knowledge, no
similar characterizations to Theorem \ref{char Shapley-vp} exist in the
literature.


\section{Discussion}

We have studied in this paper an alternative index to measure the importance
of artists in streaming platforms. This index is based on the Shapley value
of the TU game associated to streaming platforms. Our axiomatic analysis
shows that there are pros and cons using this index with respect to the
alternative standard indices. The cons are captured by the axiom of pairwise homogeneity (which is satisfied by both the pro-rata and user-centric indices, but not by the Shapley index). The pros depend on the specific comparison. When comparing with the pro-rata index, they refer to the axioms of reasonable lower bound, symmetry on fans, non-unilateral manipulability, equal global impact of users, and clik-fraud proofness (all of them satisfied by the Shapley index, but not by the pro-rata index). When comparing with the user-centric index, they refer to non-unilateral manipulability, symmetry on fans, and equal impact of artists (all of them satisfied by the Shapley index, but not by the user-centric index).

\bigskip

We conclude stressing that when it comes to measuring the importance of artists in streaming platforms,
two aspects are relevant: the users streaming content from
the artist and the streaming times. Although related, both aspects are different. 
We believe that pro-rata manages well streaming times
but not the number of users. Shapley manages well the number of users but
not streaming times. User-centric manages well both aspects. The next examples illustrate these comments.

Suppose first that three users decided to join a platform to stream their
favorite artists while commuting every day. Artist $1$ is the favorite of
user $a$, whereas artist $2$ is the favorite of users $b$ and $c$. Besides,
the commuting time of users $b$ and $c$ coincide but the commuting time of
user $a$ is twice that. In this case, it seems reasonable that the amount
paid by user $a$ goes to artist $1$, whereas the amount paid by users $b$
and $c$ goes to artist $2$. This is what the user-centric and Shapley
indices would recommend, as shown next.

\begin{example}
\label{ex 2,3,a}Let $N=\left\{ 1,2\right\} ,$ $M=\left\{ a,b,c\right\} $ and 
\begin{equation*}
t=\left( 
\begin{array}{ccc}
200 & 0 & 0 \\ 
0 & 100 & 100%
\end{array}
\right) .
\end{equation*}

Thus, 
\begin{equation*}
\begin{tabular}{cccc}
& $P_{i}\left( N,M,t\right) $ & $U_{i}\left( N,M,t\right) $ & $Sh_{i}\left(
N,M,t\right) $ \\ 
1 & 1.5 & 1 & 1 \\ 
2 & 1.5 & 2 & 2%
\end{tabular}%
\end{equation*}

Notice that user-centric and Shapley propose the \textquotedblleft
reasonable\textquotedblright\ allocation whereas pro-rata gives the same
amount to both artists.
\end{example}

Consider now a modified scenario in which the three users have the same
commuting time and they all like artist 2 twice more than artist 1. In this
case, it seems reasonable that artist 2 receives twice the amount received
by artist 1. This is what the user-centric and pro-rata indices would
recommend, as shown next.

\begin{example}
\label{ex 2,3,b}Let $N=\left\{ 1,2\right\} ,$ $M=\left\{ a,b,c\right\} $ and 
\begin{equation*}
t^{\prime }=\left( 
\begin{array}{ccc}
100 & 100 & 100 \\ 
200 & 200 & 200%
\end{array}
\right).
\end{equation*}

Thus, 
\begin{equation*}
\begin{tabular}{cccc}
& $P_{i}\left( N,M,t^{\prime }\right) $ & $U_{i}\left( N,M,t^{\prime
}\right) $ & $Sh_{i}\left( N,M,t^{\prime }\right) $ \\ 
1 & 1 & 1 & 1.5 \\ 
2 & 2 & 2 & 1.5%
\end{tabular}%
\end{equation*}

Notice that pro-rata and user-centric propose the \textquotedblleft
reasonable\textquotedblright\ allocation whereas Shapley gives the same
amount to both artists.
\end{example}

\section*{Appendix}



\subsection*{Independence of the axioms in Theorem \protect\ref{char Shapley}}

Let $I^{1}$ be defined as follows. For each $\left( N,M,t\right)\in\mathcal{P%
}$, let $N^{\prime }=\left\{ i\in N:F_{i}\neq \varnothing \right\}$. Then,
for each $i\in N,$ 
\begin{equation*}
I_{i}^{1}\left( N,M,t\right) =\left\{ 
\begin{tabular}{ll}
$\frac{m}{\left\vert N^{\prime }\right\vert }$ & if $i\in N^{\prime }$ \\ 
0 & otherwise.%
\end{tabular}
\right. .
\end{equation*}

$I^{1}$ satisfies all axioms in Theorem \ref{char Shapley} except for
additivity.

Let $I^{2}$ be defined as follows. For each $\left( N,M,t\right) \in 
\mathcal{P}$ and each $i\in N,$ 
\begin{equation*}
I_{i}^{2}\left( N,M,t\right) =\frac{m}{n}.
\end{equation*}

$I^{2}$ satisfies all axioms in Theorem \ref{char Shapley} except for
reasonable lower bound 

Let $I^{3}$ be defined as follows. Let $\left( w_{j}\right) _{j\in M}$ be a
vector of positive numbers. For each $\left( N,M,t\right) \in \mathcal{P}$
and each $i\in N,$ 
\begin{equation*}
I_{i}^{3}\left( N,M,t\right) =\sum_{j\in M:i\in L^{j}}\frac{w_{j}}{%
\left\vert L_{j}\left( N,M,t\right) \right\vert }.
\end{equation*}

$I^{3}$ satisfies all axioms in Theorem \ref{char Shapley} except for equal
global impact of users.

The user-centric rule satisfies all axioms in Theorem \ref{char Shapley}
except for symmetry on fans.

\subsection*{Independence of the axioms in Theorem \protect\ref{char Shapley
song lengh}}

$I^{1}$ satisfies all axioms in Theorem \ref{char Shapley song lengh} except
for additivity.

$I^{2}$ satisfies all axioms in Theorem \ref{char Shapley song lengh} except
for reasonable lower bound 

$I^{3}$ satisfies all axioms in Theorem \ref{char Shapley song lengh} except
for equal global impact of users.

Let $I^{4}$ be defined as follows. Let $\left( w_{i}\right) _{i\in N}$ be a
vector of positive numbers. For each $\left( N,M,t\right) \in \mathcal{P}$
and each $i\in N,$ 
\begin{equation*}
I_{i}^{4}\left( N,M,t\right) =\sum_{j\in M:i\in L^{j}}\frac{w_{i}}{%
\sum\limits_{i^{\prime }\in \left\vert L_{j}\left( N,M,t\right) \right\vert
}w_{i^{\prime }}}.
\end{equation*}

$I^{4}$ satisfies all axioms in Theorem \ref{char Shapley song lengh} except
for order preservation.

The user-centric rule satisfies all axioms in Theorem \ref{char Shapley song
lengh} except for non-unilateral manipulability.

\subsection*{Independence of the axioms in Theorem \protect\ref{char Shapley-vp}%
}

\label{indep char Shapley}

$I^{1}$ defined above satisfies all axioms in Theorem \ref{char Shapley-vp}
except for additivity.

$I^{2}$ satisfies all axioms in Theorem \ref{char Shapley-vp} except for
reasonable lower bound.

$I^{3}$ satisfies all axioms in Theorem \ref{char Shapley-vp} except for
equal global impact of users.

The user-centric rule satisfies all axioms in Theorem \ref{char Shapley-vp}
except for equal impact of artists.


\end{document}